\documentclass[twoside]{article}%
\usepackage{amssymb}
\usepackage{amsfonts}
\usepackage{amsmath}
\usepackage{graphicx}%
\providecommand{\U}[1]{\protect\rule{.1in}{.1in}}
\newtheorem{theorem}{Theorem}

\newtheorem{lemma}[theorem]{Lemma}

\newtheorem{remark}[theorem]{Remark}

\newenvironment{proof}[1][Proof]{\noindent\textbf{#1.} }{\ \rule{0.5em}{0.5em}}
\begin{document}

\title{\textbf{Analytically Periodic Solutions to the }$3$\textbf{-dimensional
Euler-Poisson Equations of Gaseous Stars with Negative Cosmological Constant}}
\author{Y\textsc{uen} M\textsc{anwai\thanks{E-mail address: nevetsyuen@hotmail.com }}\\\textit{Department of Applied Mathematics,}\\\textit{The Hong Kong Polytechnic University,}\\\textit{Hung Hom, Kowloon, Hong Kong} }
\date{Revised Date: 11-June-2009}
\maketitle

\begin{abstract}
By the extension of the 3-dimensional analytical solutions of Goldreich and
Weber \cite{GW} with adiabatic exponent $\gamma=4/3$, to the (classical)
Euler-Poisson equations without cosmological constant, the self-similar
(almost re-collapsing) time-periodic solutions with negative cosmological
constant ($\Lambda<0$) are constructed. The solutions with time-periodicity
are novel. On basing these solutions, the time-periodic and almost
re-collapsing model is conjectured, for some gaseous stars.


\end{abstract}

\section{Introduction}

The evolution of a self-gravitating fluid, such as gaseous stars, can be
formulated by the isentropic Euler-Poisson systems of the following form:
\begin{equation}
\left\{
\begin{array}
[c]{rl}%
{\normalsize \rho}_{t}{\normalsize +\nabla\bullet(\rho\vec{u})} &
{\normalsize =}{\normalsize 0,}\\
{\normalsize (\rho\vec{u})}_{t}{\normalsize +\nabla\bullet(\rho\vec{u}%
\otimes\vec{u})+\nabla}P & {\normalsize =}{\normalsize -\rho\nabla\Phi,}\\
{\normalsize \Delta\Phi(t,\vec{x})} & {\normalsize =\alpha(N)}%
{\normalsize \rho-\Lambda,}%
\end{array}
\right.  \label{Euler-Poisson}%
\end{equation}
where $\alpha(N)$ is a constant related to the unit ball in $R^{N}$:
$\alpha(2)=2\pi$ and $\alpha(3)=4\pi$. And as usual, $\rho=\rho(t,\vec{x})$
and $\vec{u}=\vec{u}(t,\vec{x})\in\mathbf{R}^{N}$ are the density, the
velocity respectively. $P=P(\rho)$\ is the pressure function. While $\Lambda$
is the cosmological constant, for $\Lambda>0$, the space is open; for
$\Lambda<0$, the space is closed; for $\Lambda=0$ the space is flat.

In the above systems, the self-gravitational potential field $\Phi=\Phi
(t,\vec{x})$\ is determined by the density $\rho$ itself, through the Poisson
equation (\ref{Euler-Poisson})$_{3}$. For $N=3$, (\ref{Euler-Poisson}) are the
classical (non-relativistic) descriptions of a galaxy, in astrophysics. See
\cite{BT}, \cite{C}, \cite{KW}, for details about the systems.

The $\gamma$-law can be applied on the pressure term $P(\rho)$, i.e.%
\begin{equation}
{\normalsize P}\left(  \rho\right)  {\normalsize =K\rho}^{\gamma}\doteq
\frac{{\normalsize \rho}^{\gamma}}{\gamma},\label{gamma}%
\end{equation}
which is a common hypothesis. The constant $\gamma=c_{P}/c_{v}\geq1$, where
$c_{P}$, $c_{v}$\ are the specific heats per unit mass under constant pressure
and constant volume respectively, is the ratio of the specific heats, that is,
the adiabatic exponent in (\ref{gamma}). In particular, the fluid is called
isothermal if $\gamma=1$.\newline Historically in astrophysics for zero
constant $(\Lambda=0)$ background, for the $3$-dimensional space, the
hydrostatic equilibrium specified by $u=0$, was studied. According to
\cite{C}, the ratio between the core density $\rho(0)$ and the mean density
$\overset{\_}{\rho}$ for $6/5<\gamma<2$\ is given by%
\begin{equation}
\frac{\overset{\_}{\rho}}{\rho(0)}=\left(  \frac{-3}{z}\dot{y}\left(
z\right)  \right)  _{z=z_{0}},
\end{equation}
where $y$\ is the solution of the Lane-Emden equation with $n=1/(\gamma-1)$
for $\gamma>1:$%
\begin{equation}
\ddot{y}(z)+\dfrac{2}{z}\dot{y}(z)+y(z)^{n}=0,\text{ }y(0)=\alpha>0,\text{
}\dot{y}(0)=0,\text{ }n=\frac{1}{\gamma-1},
\end{equation}
and $z_{0}$\ is the first zero of $y(z_{0})=0$. We can solve the Lane-Emden
equation analytically only for%
\begin{equation}
y_{anal}(z)\doteq\left\{
\begin{array}
[c]{ll}%
1-\dfrac{1}{6}z^{2}, & n=0;\\
\dfrac{\sin z}{z}, & n=1;\\
\dfrac{1}{\sqrt{1+z^{2}/3}}, & n=5,
\end{array}
\right.
\end{equation}
and for the other values, only numerical values can be obtained. It can be
shown that for $n<5$, the radius of polytropic models is finite; for $n\geq5$,
the radius is infinite.\newline As we are interested in the solutions in
radial symmetry, the Poisson equation (\ref{Euler-Poisson})$_{3}$ is
transformed to%
\begin{equation}
{\normalsize r^{N-1}\Phi}_{rr}\left(  {\normalsize t,x}\right)  +\left(
N-1\right)  r^{N-2}\Phi_{r}{\normalsize =}\alpha\left(  N\right)
({\normalsize \rho-\Lambda)r^{N-1},}%
\end{equation}%
\begin{equation}
\Phi_{r}=\frac{\alpha\left(  N\right)  }{r^{N-1}}\int_{0}^{r}(\rho
(t,s)-\Lambda)s^{N-1}ds.
\end{equation}
For the analytical collapsing solutions, the well-known results are
constructed by Goldreich and Weber \cite{GW} in 1980 for $\gamma=4/3$ in
$R^{3}$. The analytical collapsing solutions without compact support for
$\gamma=1$ in $R^{2}$ are constructed by Yuen \cite{Y1} in 2008.

For simplicity, we take the constant $\Lambda=-\frac{3}{4\pi}$ for $\Lambda
<0$. For the background knowledge for the Euler-Poisson equations with
cosmological constant $\Lambda$, the interested readers may refer \cite{FT},
\cite{PR}.

Our main results are able to obtain the analytically time-periodic solutions
for the Euler-Poisson equations with the negative cosmological constant
($\Lambda<0$) in spherical symmetry:%
\begin{equation}
\vec{u}=\frac{\vec{x}}{r}V(t,r),
\end{equation}
where the diameter $r:=\left(  \sum_{i=1}^{3}x_{i}^{2}\right)  ^{1/2}$ in
$R^{3}$,%
\begin{equation}
\left\{
\begin{array}
[c]{rl}%
\rho_{t}+V\rho_{r}+\rho V_{r}+{\normalsize \dfrac{2}{r}\rho V} &
{\normalsize =0,}\\
\rho\left(  V_{t}+VV_{r}\right)  +\frac{\partial}{\partial r}(K\rho_{{}}%
^{4/3}) & {\normalsize =-}\dfrac{4\pi\rho}{r^{2}}\int_{0}^{r}(\rho
(t,s)+\frac{3}{4\pi})s^{2}ds,
\end{array}
\right.  \label{gamma=1}%
\end{equation}
with the extension of the analytical solutions of Goldreich and Weber
\cite{GW} in the theorem:

\begin{theorem}
\label{thm:1}For the $3$-dimensional Euler-Poisson equations in spherical
symmetry (\ref{gamma=1}), there exists a family of solutions, for $\gamma
=4/3$:%
\begin{equation}
\left\{
\begin{array}
[c]{c}%
\rho(t,r)=\left\{
\begin{array}
[c]{c}%
\frac{1}{a(t)^{3}}y(\frac{r}{a(t)})^{3},\text{ for }r<a(t)Z_{\nu};\\
0,\text{ for }a(t)Z_{\nu}\leq r.
\end{array}
\right.  \text{, }V{\normalsize (t,r)=}\frac{\overset{\cdot}{a}(t)}%
{a(t)}{\normalsize r;}\\
\overset{\cdot\cdot}{a}(t){\normalsize =}\frac{-\lambda}{a(t)^{2}}-a(t),\text{
}{\normalsize a(0)=a}_{0}>0{\normalsize ,}\text{ }\overset{\cdot}%
{a}(0){\normalsize =a}_{1};\\
\overset{\cdot\cdot}{y}(z){\normalsize +}\frac{2}{z}\overset{\cdot}%
{y}(z){\normalsize +}\frac{\pi}{K}{\normalsize y(z)}^{3}{\normalsize =}%
\mu,\text{ }y(0)=\alpha>0,\text{ }\overset{\cdot}{y}(0)=0,
\end{array}
\right.  \label{solution1}%
\end{equation}
where $K>0$, $\mu=3\lambda/(4K)$, and the finite $Z_{\nu}$ is the first zero
point of $y(z)$.\newline In particular,\newline(1) $\lambda=0$ and $a_{1}=0$,
the solutions collapse in the finite time $t=\pi/2;$\newline(2) $\lambda<0$,
the solutions are non-trivially time-periodic, except the case with
$a_{0}=\sqrt[3]{\lambda}$ and $a_{1}=0$.;\newline(3) $\lambda>0$ and
$a_{1}\leq0$, the solutions collapse in a finite time $T$.
\end{theorem}

\section{Time-Periodic Solutions}

In this section, before presenting the proof of Theorem \ref{thm:1}, we
prepare some lemmas. First, we obtain a general class of solutions for the
mass equation in spherical symmetry (\ref{gamma=1})$_{1}$.

\begin{lemma}
\label{lem:generalsolutionformasseq copy(1)}For the 3-dimensional equation of
conservation of mass in spherical symmetry:
\begin{equation}
\rho_{t}+V\rho_{r}+\rho V_{r}+\frac{2}{r}\rho V=0, \label{eq1122}%
\end{equation}
there exist solutions,%
\begin{equation}
\rho(t,r)=\frac{f(r/a(t))}{a(t)^{3}},\text{ }V{\normalsize (t,r)=}%
\frac{\overset{\cdot}{a}(t)}{a(t)}{\normalsize r,} \label{eq1123}%
\end{equation}
with the form $f\geq0\in C^{1}$ and $a(t)>0\in C^{1}.$
\end{lemma}

\begin{proof}
We just plug (\ref{eq1123}) into (\ref{eq1122}). Then
\begin{align}
&  \rho_{t}+V\rho_{r}+\rho V_{r}+\frac{2}{r}\rho V\\
&  =\frac{-3\overset{\cdot}{a}(t)f(r/a(t))}{a(t)^{4}}-\frac{\overset{\cdot}%
{a}(t)r\overset{\cdot}{f}(r/a(t))}{a(t)^{5}}\\
&  +\frac{\overset{\cdot}{a}(t)r}{a(t)}\frac{\overset{\cdot}{f}(r/a(t))}%
{a(t)^{4}}+\frac{f(r/a(t))}{a(t)^{3}}\frac{\overset{\cdot}{a}(t)}{a(t)}%
+\frac{2}{r}\frac{f(r/a(t))}{a(t)^{3}}\frac{\overset{\cdot}{a}(t)}{a(t)}r\\
&  =0.
\end{align}
The proof is completed.
\end{proof}

The following core lemma is needed to show the cyclic phenomena of the
solutions (\ref{solution1}).

\begin{lemma}
\label{lemma22}For the ordinary differential equation,%
\begin{equation}
\ddot{a}(t)-\frac{\lambda}{a(t)^{2}}+a(t)=0,\text{ }a(0)=a_{0}>0,\text{ }%
\dot{a}(0)=a_{1},\label{eq124}%
\end{equation}
with $\lambda>0$, there exist non-trivially periodic solutions except the case
with $a_{0}=\sqrt[3]{\lambda}$ and $a_{1}=0$.
\end{lemma}

\begin{proof}
To the equation (\ref{eq124}), multiply $\dot{a}(t)$ and then integrate it:%
\begin{equation}
\frac{\dot{a}(t)^{2}}{2}+\frac{\lambda}{a(t)}+\frac{a(t)^{2}}{2}%
=\theta,\label{rt1}%
\end{equation}
with the constant $\theta=\frac{a_{1}^{2}}{2}+\frac{\lambda}{a_{0}}%
+\frac{a_{0}^{2}}{2}>0$.\newline We define the kinetic energy:
\begin{equation}
F_{kin}:=\frac{\dot{a}(t)^{2}}{2},
\end{equation}
and the potential energy:%
\begin{equation}
F_{pot}=\frac{\lambda}{a(t)}+\frac{a(t)^{2}}{2}.
\end{equation}
The total energy is conserved:%
\begin{equation}
\frac{d}{dt}(F_{kin}+F_{pot})=0.
\end{equation}
Observe the potential energy has only one global minimum at $a_{\min}%
=\sqrt[3]{\lambda}$, for $a(t)\in(0,+\infty)$. Therefore, by the classical
energy method for conservative systems (in Section 4.3 of \cite{LS}), the
solutions have the closed trajectory. We calculate the time for traveling the
closed orbit:%
\begin{equation}
T=\int_{0}^{t}dt=2\int_{a_{\max}}^{a_{\min}}\frac{-da(t)}{\sqrt{2[\theta
-(\frac{\lambda}{a(t)}+\frac{a(t)^{2}}{2})]}}=2\int_{a_{\min}}^{a_{\max}}%
\frac{da(t)}{\sqrt{2[\theta-(\frac{\lambda}{a(t)}+\frac{a(t)^{2}}{2})]}%
},\label{hk1}%
\end{equation}
where $a_{\min}=\underset{t\geq0}{\inf}(a(t))$ and $a_{\max}=\underset{t\geq
0}{\sup}(a(t)).$\newline Let $G(t):=\theta-(\frac{\lambda}{a(t)}%
+\frac{a(t)^{2}}{2})$, $G_{0}:=\theta-(\frac{\lambda}{a_{\min}+\epsilon}%
+\frac{(a_{\min}+\epsilon)^{2}}{2})$ $>0$ and $G_{1}:=\theta-(\frac{\lambda
}{a_{\min}+\epsilon}+\frac{(a_{\min}+\epsilon)^{2}}{2})>0$. Except the case
with $a_{0}=\sqrt[3]{\lambda}$ and $a_{1}=0$, the equation (\ref{hk1})
becomes
\begin{align}
T &  =\int_{a_{\min}}^{a_{\min}+\epsilon}\frac{2da(t)}{\sqrt{2[\theta
-(\frac{\lambda}{a(t)}+\frac{a(t)^{2}}{2})]}}+\int_{a_{\min}+\epsilon
}^{a_{\max}-\epsilon}\frac{2da(t)}{\sqrt{2[\theta-(\frac{\lambda}{a(t)}%
+\frac{a(t)^{2}}{2})]}}\\
&  +\int_{a_{\max}-\epsilon}^{a_{\max}}\frac{2da(t)}{\sqrt{2[\theta
-(\frac{\lambda}{a(t)}+\frac{a(t)^{2}}{2})]}}\\
&  \leq\underset{a_{\min}\leq a(t)\leq a_{\min}+\epsilon}{\sup}\left\vert
\frac{1}{\frac{\lambda}{a(t)^{2}}-a(t)}\right\vert \int_{0}^{G_{0}}\frac
{\sqrt{2}dG(t)}{\sqrt{G(t)}}+\int_{a_{\min}+\epsilon}^{a_{\max}-\epsilon}%
\frac{2da(t)}{\sqrt{2[\theta-(\frac{\lambda}{a(t)}+\frac{a(t)^{2}}{2})]}}\\
&  +\underset{a_{\max}-\epsilon\leq a(t)\leq a_{\max}}{\sup}\left\vert
\frac{1}{\frac{\lambda}{a(t)^{2}}-a(t)}\right\vert \int_{0}^{G_{1}}\frac
{\sqrt{2}dG(t)}{\sqrt{G(t)}}\\
&  =\underset{a_{\min}\leq a\leq a_{\min}+\epsilon}{\sup}\left\vert \frac
{1}{\frac{\lambda}{a(t)^{2}}-a(t)}\right\vert \frac{\sqrt{G_{0}}}{\sqrt{2}%
}+\int_{a_{\min}+\epsilon}^{a_{\max}-\epsilon}\frac{2da(t)}{\sqrt
{2[\theta-(\frac{\lambda}{a(t)}+\frac{a(t)^{2}}{2})]}}\\
&  +\underset{a_{\max}-\epsilon\leq a(t)\leq a_{\max}}{\sup}\left\vert
\frac{1}{\frac{\lambda}{a(t)^{2}}-a(t)}\right\vert \frac{\sqrt{G_{1}}}%
{\sqrt{2}}\\
&  <\infty.
\end{align}
Therefore, the solutions to the differential equation (\ref{eq124}) are
time-periodic.\newline The proof is completed
\end{proof}

In the solutions (\ref{solution1}), due to the properties of the modified
Emden equation (\ref{solution1})$_{2}$, the collapsing phenomena are observed.
The following lemma can be proved immediately.

\begin{lemma}
\label{lemma33}For the ordinary differential equation,%
\begin{equation}
\ddot{a}(t)+\frac{\lambda}{a(t)^{2}}=-a(t),\text{ }0<a(0)=a_{0},\ \dot
{a}(0)=a_{1}\leq0,\text{ } \label{eqo1}%
\end{equation}
with $\lambda>0$, \newline there exists a finite time $T$ such that
$\underset{t\rightarrow T^{-}}{\lim}a(t)=0.$
\end{lemma}

\begin{proof}
If the claim is not true, for all $t\geq0$, we have $a(t)>0$.\newline But the
ordinary differential equation (\ref{eqo1}) becomes%
\begin{equation}
\overset{\cdot\cdot}{a}(t){\normalsize =}\frac{-1}{a(t)^{2}}-a(t)<0.
\end{equation}
This shows $\dot{a}(t)$ is decreasing function: $\dot{a}(t)<\dot{a}(t_{1})$
for all $t>t_{1}>0$,%
\begin{equation}
\dot{a}(t)\leq\dot{a}(t_{1})<0.
\end{equation}
Thus, the solution is bounded by%
\begin{equation}
a(t)=\int_{t_{1}}^{t}\dot{a}(s)ds+a(t_{1})\leq\dot{a}(t_{1})t+a_{0}.
\end{equation}
After a sufficient large time, there exists a finite time $T$ such that
$\underset{t\rightarrow T^{-}}{\lim}a(t)=0$. A contraction is met.\newline The
proof is completed.
\end{proof}

Here we are ready to give proof of Theorem \ref{thm:1}.

\begin{proof}
[Proof of Theorem \ref{thm:1}]From Lemma
\ref{lem:generalsolutionformasseq copy(1)}, we can easily check that
(\ref{solution1}) satisfy (\ref{gamma=1})$_{1}$. We plug the solutions
(\ref{solution1}), into the momentum equation (\ref{gamma=1})$_{2}$:%
\begin{align}
&  \rho(V_{t}+VV_{r})+K\frac{\partial}{\partial r}\rho^{4/3}+\frac{4\pi\rho
}{r^{2}}%
{\displaystyle\int\limits_{0}^{r}}
(\rho(t,s)-\Lambda)s^{2}ds\\
&  =\rho\frac{\overset{\cdot\cdot}{a}(t)}{a(t)}r+4K\left(  \frac{y(\frac
{r}{a(t)})^{3}}{a(t)^{3}}\right)  ^{1/3}\frac{y(\frac{r}{a(t)})^{2}%
\overset{\cdot}{y}(\frac{r}{a(t)})}{a(t)^{4}}+\frac{4\pi\rho}{r^{2}}%
{\displaystyle\int\limits_{0}^{r}}
(\rho(t,s)-\Lambda)s^{2}ds\\
&  =\rho\left[  \frac{\frac{-\lambda}{a(t)^{2}}+\frac{4\pi\Lambda}{3}%
a(t)}{a(t)}r\right]  +4K\frac{y(\frac{r}{a(t)})^{3}}{a(t)^{3}}\frac
{\overset{\cdot}{y}(\frac{r}{a(t)})}{a(t)^{2}}+\frac{4\pi\rho}{r^{2}a(t)^{2}}%
{\displaystyle\int\limits_{0}^{r}}
y^{3}(\frac{s}{a(t)})s^{2}ds-\frac{4\pi\Lambda\rho}{r^{2}}%
{\displaystyle\int\limits_{0}^{r}}
s^{2}ds\\
&  =\rho\frac{-\lambda r}{a(t)^{3}}+\frac{4\pi\Lambda}{3}\rho r+4K\rho
\frac{\overset{\cdot}{y}(\frac{r}{a(t)})}{a(t)^{2}}+\frac{4\pi\rho}%
{r^{2}a(t)^{3}}%
{\displaystyle\int\limits_{0}^{r}}
y(\frac{s}{a(t)})^{3}s^{2}ds-\frac{4\pi\Lambda}{3}\rho r\\
&  =\frac{\rho}{a(t)^{2}}\left[  -\frac{\lambda}{a(t)}r{\normalsize +4}%
K\overset{\cdot}{y}(\frac{r}{a(t)})+\frac{4\pi}{r^{2}a(t)}%
{\displaystyle\int\limits_{0}^{r}}
y(\frac{s}{a(t)})^{3}s^{2}ds\right]  \\
&  =\frac{\rho}{a(t)^{2}}\left[  -\frac{\lambda}{a(t)}r+4K\overset{\cdot}%
{y}(\frac{r}{a(t)})+\frac{\alpha(N)}{(\frac{r}{a(t)})^{2}}%
{\displaystyle\int\limits_{0}^{r/a(t)}}
y(s)^{3}s^{2}ds\right]  \\
&  =\frac{\rho}{a(t)^{2}}Q\left(  \frac{r}{a(t)}\right)  .
\end{align}
Here, we use the property of $a(t)$:%
\begin{equation}
\overset{\cdot\cdot}{a}(t)=\frac{-\lambda}{a(t)^{2}}+\frac{4\pi\Lambda}%
{3}a(t)=\frac{-\lambda}{a(t)^{2}}-a(t),
\end{equation}
and denote%
\begin{equation}
Q(\frac{r}{a(t)}):={\normalsize Q(x)=-\lambda x+4}K\overset{\cdot}%
{y}(x){\normalsize +}\frac{4\pi}{x^{2}}%
{\displaystyle\int\limits_{0}^{x}}
y(s)^{3}s^{2}ds{\normalsize .}%
\end{equation}
Differentiate $Q(x)$\ with respect to $x$,%
\begin{align}
\overset{\cdot}{Q}(x) &  =-{\normalsize \lambda+4K}\overset{\cdot\cdot}%
{y}(x){\normalsize +}4\pi y(x)^{3}-\frac{2\cdot4\pi}{x^{3}}%
{\displaystyle\int\limits_{0}^{x}}
y(s)^{3}s^{2}{\normalsize ds}\\
&  =-\frac{2}{x}\left[  \lambda x+4K\overset{\cdot}{y}(x)-K\mu x+\frac{4\pi
}{x^{2}}%
{\displaystyle\int\limits_{0}^{x}}
y(s)^{3}s^{2}ds\right]  \\
&  =-\frac{2}{x}Q(x).
\end{align}
With $\underset{x\rightarrow0^{+}}{\lim}Q(x)=Q(0)=0$, this implies that
$Q(x)=0$.\newline The above result is due to the fact that we choose the
Lane-Emden equation:%
\begin{equation}
\left\{
\begin{array}
[c]{c}%
\overset{\cdot\cdot}{y}(z){\normalsize +}\frac{2}{z}\overset{\cdot}%
{y}(z){\normalsize +}\frac{\pi}{K}{\normalsize y(z)}^{3}{\normalsize =}%
\mu{\normalsize ,}\text{ }\mu=\frac{3\lambda}{4K},\\
{\normalsize y(0)=\alpha>0,}\text{ }\overset{\cdot}{y}(0){\normalsize =0.}%
\end{array}
\right.  \label{eq11122}%
\end{equation}
We have shown that there exists the family of the solutions with compact
support, in Theorem \ref{thm:1}, with the well-known results about the
Lane-Emden equation (\ref{eq11122}) \cite{C} \cite{KW}.\newline On the other
hand, for $\Lambda=-\frac{3}{4}\pi<0$, it is clear for $\lambda=0$ in the
theorem is true. For the linear ordinary differential equation:%
\begin{equation}
\ddot{a}(t)+a(t)=0,\text{ }a(0)=a_{0}>0,\dot{a}(0)=0,
\end{equation}
the solution is
\begin{equation}
a(t)=a_{0}\cos t\text{.}%
\end{equation}
The function $a(\frac{\pi}{2})$ achieves zero and the solutions collapse in
the finite time $\pi/2$.\newline By using Lemma \ref{lemma22} about $a(t)$, we
have , for $\lambda<0$, the periodic solutions are non-trivial except the case
with $a_{0}=\sqrt[3]{\lambda}$ and $a_{1}=0$.\newline By using Lemma
\ref{lemma33}, we have, for $\lambda>0$ and $a_{1}\leq0$ the solutions
collapse in a finite time $T$.\newline This completes the proof.
\end{proof}

\begin{remark}
If we consider the system with frictional damping,
\begin{equation}
\left\{
\begin{array}
[c]{rl}%
{\normalsize \rho}_{t}{\normalsize +\nabla\bullet(\rho\vec{u})} &
{\normalsize =}{\normalsize 0,}\\
{\normalsize (\rho\vec{u})}_{t}{\normalsize +\nabla\bullet(\rho\vec{u}%
\otimes\vec{u})+\beta\rho\vec{u}+\nabla}P & {\normalsize =}{\normalsize -\rho
\nabla\Phi,}\\
{\normalsize \Delta\Phi(t,\vec{x})} & {\normalsize =\alpha(N)}%
{\normalsize \rho-}\Lambda{\normalsize ,}%
\end{array}
\right.
\end{equation}
with the constant $\beta>0$,\newline the corresponding ordinary differential
equation for the solutions are:%
\begin{equation}
\overset{\cdot\cdot}{a}(t)+\beta\dot{a}(t){\normalsize =}\frac{-\lambda
}{a(t)^{2}}-a(t),\text{ }{\normalsize a(0)=a}_{0}>0{\normalsize ,}\text{
}\overset{\cdot}{a}(0){\normalsize =a}_{1},
\end{equation}
\cite{Y}.\newline For $\lambda<0$, the function $a(t)$ oscillates around the
minimum $a_{\min}=\sqrt[3]{\lambda}$ of the potential energy except the case
with $a_{0}=\sqrt[3]{\lambda}$ and $a_{1}=0$. And it is asymptotically stable.
\end{remark}

\begin{remark}
In general, for the Euler-Piosson equations in $R^{N}$, it is clear to have
the corresponding solutions:\newline for $N=2$ and $\gamma=1$,
\begin{equation}
\left\{
\begin{array}
[c]{c}%
\rho(t,r)=\dfrac{1}{a^{2}(t)}e^{y\left(  r/a(t)\right)  }\text{,
}V{\normalsize (t,r)=}\dfrac{\dot{a}(t)}{a(t)}{\normalsize r;}\\
\ddot{a}(t){\normalsize =}\dfrac{-\lambda}{a(t)}-a(t),\text{ }%
{\normalsize a(0)=a}_{0}>0{\normalsize ,}\text{ }\dot{a}(0){\normalsize =a}%
_{1};\\
\ddot{y}(z){\normalsize +}\dfrac{1}{z}\dot{y}(z){\normalsize +\dfrac{2\pi}%
{K}e}^{y(z)}{\normalsize =\mu,}\text{ }y(0)=\alpha,\text{ }\dot{y}(0)=0,
\end{array}
\right.  \label{solution 3}%
\end{equation}
where $\mu=2\lambda/K$, \cite{Y1};\newline for $N\geq3$ and $\gamma
=(2N-2)/N$,
\begin{equation}
\left\{
\begin{array}
[c]{c}%
\rho(t,r)=\left\{
\begin{array}
[c]{c}%
\dfrac{1}{a^{N}(t)}y(\frac{r}{a(t)})^{N/(N-2)},\text{ for }r<a(t)Z_{\mu};\\
0,\text{ for }a(t)Z_{\mu}\leq r.
\end{array}
\right.  \text{, }V{\normalsize (t,r)=}\dfrac{\dot{a}(t)}{a(t)}%
{\normalsize r,}\\
\ddot{a}(t){\normalsize =}\dfrac{-\lambda}{a^{N-1}(t)}-a(t),\text{
}{\normalsize a(0)=a}_{0}>0{\normalsize ,}\text{ }\dot{a}(0){\normalsize =a}%
_{1},\\
\ddot{y}(z){\normalsize +}\dfrac{N-1}{z}\dot{y}(z){\normalsize +}\dfrac
{\alpha(N)}{(2N-2)K}{\normalsize y(z)}^{N/(N-2)}{\normalsize =\mu,}\text{
}y(0)=\alpha>0,\text{ }\dot{y}(0)=0,
\end{array}
\right.  \label{solution2}%
\end{equation}
where $\mu=[N(N-2)\lambda]/(2N-2)K$ and the finite $Z_{\mu}$ is the first zero
of $y(z)$ \cite{DXY}.\newline With $\lambda<0$, the solutions \ref{solution 3}
and \ref{solution2} are time-periodic for the negative cosmological constant.
\end{remark}

\begin{remark}
Under some initial conditions for $\Lambda>0$, the collapsing solutions are
obtained. The case is similar to the system without cosmological constant.
\end{remark}

\section{Discussion}

By the separation method, the solutions are constructed in its close relatives
(the Euler-Poisson equations with frictional damping, Navier-Stokes, Euler
equations and the ones with frictional damping term) \cite{Y} and \cite{Y2}.
However, we emphasis on that it is novel to have the time-periodic pattern in
the solution structure with $\Lambda<0$.

On the other hand, in the statement (2) of the theorem, with some initial
conditions, for example, $0<a(0)=\epsilon<<1$ and $\dot{a}(0)=0$, the density
becomes%
\begin{equation}
\rho(0,0)=\frac{\alpha^{3}}{\epsilon^{3}}>>M_{0}\text{,}%
\end{equation}
where $M_{0}$ is an arbitrary constant.\newline It may provide the possible
evolutionary model with $\Lambda<0$ constant in mathematical theory, that the
universe expands and then almost re-collapses (The density at the origin can
be greater than any given constant $\rho(T,0)>>M_{0}$), in the finite time
periodically. The solutions may be like a black hole (in the sense,
$\rho(T,0)=+\infty$ with a finite time $T$). The black hole seems to be the
end of the evolution of the periodic solutions. However actually, there exists
no authentic black hole in our particular time-periodic solutions.

As the time-periodic effect is due to the negative cosmological constant in
our model, more solutions with this pattern are expected for other $\gamma$
values. Then further works may be done by the computing simulations about the
stability of the numerical solutions. If some gaseous stars (a galaxy) obey
the $\gamma$- law $(\gamma=4/3)$, it may provide an alternative explanation
about that, the time-periodic solutions coincide with the expansion segment
(the red-shift effect) in a short time. We notice that it is the significant
difference between the traditional re-collapsing model with negative
cosmological constant $\Lambda<0$ \cite{D}.

\end{document}